\pdfoutput=1
\documentclass[12pt]{article}

\usepackage[T1]{fontenc}
\usepackage{mathpazo}

\usepackage{amsmath,amsthm,bm}

\newcommand{\RomNum}[1]{%
  \textup{\uppercase\expandafter{\romannumeral#1}}%
}

\usepackage{enumitem}
\usepackage[margin=1in]{geometry}
\usepackage{comment}
\usepackage{float}
\usepackage{natbib}
\usepackage{booktabs}
\usepackage{graphicx}
\usepackage{breakcites}
\usepackage{tabularx}

\usepackage{setspace}
\usepackage{multirow}
\usepackage[flushleft]{threeparttable}

\usepackage{makecell}
\usepackage{caption}
\usepackage[dvipsnames]{xcolor}
\usepackage{pgfplots,tikz}
\pgfplotsset{compat=1.17}
\usepackage[colorlinks = true,linkcolor = blue]{hyperref}
\definecolor{darkblue}{rgb}{0, 0, 0.5}
\hypersetup{
     colorlinks = true,
     citecolor = Maroon,
     urlcolor = darkblue,
     linkcolor = darkblue
}

\usepackage{subcaption} 

\DeclareMathOperator*{\argmax}{argmax}
\DeclareMathOperator*{\argmin}{argmin}

\newtheorem{proposition}{Proposition}
\newtheorem{theorem}{Theorem}
\newtheorem{lemma}{Lemma}
\newtheorem{assumption}{Assumption}

\title{Grenander-type Density Estimation under Myerson Regularity} 
\author{Haitian Xie \\ Guanghua School of Management \\ Peking University \\ xht@gsm.pku.edu.cn}

\date{\textbf{\today}}

\begin{document}

\maketitle

\begin{abstract} 
\doublespacing
    %This study examines the density estimation of private values from second-price auctions. Instead of the standard smoothing-based density estimators, we propose a Grenander-type estimator based on a shape restriction. The shape restriction is a convexity constraint, which is an equivalent representation of the famous Myerson regularity condition in the auction theory. The proposed estimator is nonparametric and yet tunning-parameter-free. We derive the consistency and the nonnormal asymptotic distribution of the estimator under mild assumptions. We also show that the estimator achieves the minimax optimal convergence rate.
    
    This study presents a novel approach to the density estimation of private values from second-price auctions, diverging from the conventional use of smoothing-based estimators. We introduce a Grenander-type estimator, constructed based on a shape restriction in the form of a convexity constraint. This constraint corresponds to the renowned Myerson regularity condition in auction theory, which is equivalent to the concavity of the revenue function for selling the auction item. Our estimator is nonparametric and does not require any tuning parameters. Under mild assumptions, we establish the cube-root consistency and show that the estimator asymptotically follows the scaled Chernoff's distribution. Moreover, we demonstrate that the estimator achieves the minimax optimal convergence rate.

    \bigskip
    \noindent%
{\bf Keywords:} Nonparametric density estimation, Convexity constraint, Tuning-paramter-free, Chernoff's distribution, Minimax optimality.
\end{abstract}

\newpage

\section{Introduction}

%Estimating the density function of private valuations (or so-called willingness-to-pay) is an important task for econometrics and industrial organization. In this paper, we study the estimation of the density function $f$ based on an independent and identically distributed (iid) sample $(V_1,\cdots,V_n)$ of the valuations, observed from truth-revealing mechanisms such as second-price auctions on eBay, Google Ads, and Meta Ad Auction.

%The common nonparametric procedure is to use smoothing based estimators such as the kernel density estimator and the local polynomial estimator. However, there are two major issues with smoothing estimators. First, the researcher needs to choose a tuning parameter, bandwidth, which has a significant impact on the performance of the estimator. Choosing the tuning parameter optimal is a difficult task and often requires estimation of higher-order derivatives of the density function. Second, computation is difficult because the researcher needs to compute an estimate for each value in the domain of the density function. 

%In this paper, we utilize a shape restriction on the valuation distribution that is more natural than the smoothness condition in the economic setting. We develop a nonparametric density estimator for $f$ that does not require any tuning parameter. This estimator is computationally more efficient that it is piecewise defined. 

The estimation of the density of private valuations (also referred to as willingness-to-pay) is an important topic of econometrics and industrial organization. This paper focuses on estimating the density function, using an independent and identically distributed (iid) sample of valuations. These valuations are observed from truth-revealing mechanisms such as second-price auctions on digital platforms including eBay, Google Ads, and Meta Ad Auction.

Conventionally, nonparametric procedures involve the use of smoothing-based estimators such as the kernel density estimator and the local polynomial estimator. However, these estimators present three primary challenges. First, the researcher must determine a tuning parameter (bandwidth), which significantly influences the estimator's performance. This optimization is a complex task often requiring the estimation of higher-order derivatives of the density function. Second, inference procedure often requires undersmoothing because the asymptotic normal distribution achieved under the optimal convergence rate has a nonnegligible bias term. Third, computation becomes complex as an estimate needs to be computed for each value within the density function's domain.

To address these issues, this paper utilizes a shape restriction on the valuation distribution, which is more naturally suited to the economic setting than the smoothness condition. We propose a nonparametric density estimator that eliminates the need for any tuning parameter. The asymptotic distribution under the optimal convergence rate is Chernoff's distribution, which is centered at zero. This estimator is computationally more efficient as it is defined piecewise.

The shape restriction we focus on is the convexity of $(1-F)^{-1}$, where $F$ is the cumulative distribution function of the valuations. This convexity constraint is equivalent to the well-known \cite{myerson1981optimal} regularity condition, which states that the virtual valuation function needs to be increasing. As pointed out by \cite{bulow1996auctions}, Myerson regularity is also equivalent to the concavity of the revenue function from selling the item.

The asymptotic properties of our proposed estimator is nonstandard. Under mild assumptions, we derive the cube-root consistency and the non-normal asymptotic distribution of the estimator. We also demonstrate that the cube-root convergence rate is minimax optimal under the set of assumptions under consideration.

The literature on valuation density estimation has explored nonparametric methods under shape restrictions  \citep{HENDERSON2012empirical,luo2018integrated,ma2021monotonicity,pinkse2019estimation}. These papers primarily focus on first-price auctions, in which the private values are not directly observed and the monotonicity constraints they consider are different from Myerson regularity. To our knowledge, this study is the first to leverage Myerson regularity as a shape constraint in nonparametric estimation.

This paper is related to the literature on nonparametric estimation under shape restrictions. The common theme of this literature is to estimate a monotone or convex (resp. concave) function with a crucial step of taking the greatest convex minorant (resp. least concave majorant) of a preliminary estimator. These estimators are often referred to as Grenander-type estimators due to the pioneer work of \cite{grenander1956theory} for the estimation of a monotone density. Other examples of Grenander-type estimators include the estimation of a concave distribution function \citep{beare2017weak}, the estimation of a monotone hazard rate \citep{marshall1965maximum,rao1970estimation}, and isotonic regression \citep{robertson1975consistency}. General asymptotic results for nonparametric estimation under shape restrictions can be found in, for example, \cite{durot2007mathbb,durot2012limit,westling2020unified} and the book \cite{groeneboom_jongbloed_2014}. Our paper contributes to this literature by considering a new shape restriction on the distribution function based on the economically meaningful Myerson regularity condition.

\section{The model and the estimator}

We consider the second-price auction model with independent private values. An indivisible object is auctioned. We assume that there are multiple auctions which are homogeneous. All bidders are risk neutral and symmetric. Without loss of generality, we pool all the bidders together. Their private values $V_1,\cdots,V_n$ are iid draws from a common distribution $F$, which we call the value distribution. The distribution $F$ is absolutely continuous with density function $f$ and compact support $[\underline{v},\bar{v}]$.

The second-price auctions are truth revealing: it is a dominant strategy for the bidders to truthfully report their valuations. Therefore, we can assume that the researcher observes the values $V_1,\cdots,V_n$. The goal is to estimate the density function $f$ based on this sample.

It is impossible to estimate $f$ without imposing assumptions on the value distribution. The typical procedure is to assume some smoothness condition on the density function $f$ and then apply the smoothing-based methods such as the kernel density estimator. However, in the auction setting, there is a more natural restriction on the density $f$ that arises from the microeconomic theory --- \cite{myerson1981optimal} regularity condition.

\begin{assumption}[Myerson regularity] \label{ass:myerson-regularity}
    The function $v \mapsto \varphi(v) \equiv v - \frac{1-F(v)}{f(v)}$ is nondecreasing on $[\underline{v},\bar{v}]$. The function $\varphi(\cdot)$ is referred to as the virtual value function.
\end{assumption}

Myerson regularity condition is a very important condition in the mechanism design theory. Some of the most celebrated results in the theory of mechanism design require the underlying valuation distribution to be regular. For example, the second-price auction with reserve price maximizes the revenue only under the condition of regularity. 

There is an economic intuition behind Myerson regularity condition explained by \cite{bulow1996auctions}. Suppose the seller sets the price $p$ in a market with measure one buyers whose private value follows the distribution $F$. There are $1-F(p)$ buyers whose value is higher than the price $p$. These buyers will choose to purchase. The quantity sold is therefore equal to $q = 1-F(p)$. The revenue collected by the seller as a function of the quantity sold is 
\begin{align*}
    R(q) = pq = F^{-1}(1-q)q,
\end{align*}
where we assume that $F$ is strictly increasing with inverse $F^{-1}$. The marginal revenue is equal to 
\begin{align*}
    R'(q) = F^{-1}(1-q) - q/f(F^{-1}(1-q)) = p - \frac{1-F(p)}{f(p)}.
\end{align*}
Myerson regularity condition states that the marginal revenue $R'(q)$ is nondecreasing in the price $p$, which implies that it is nonincreasing in the quantity $q$. This means that Myerson regularity is equivalent to the concavity of the revenue function $R$.

From the statistical perspective, Assumption \ref{ass:myerson-regularity} imposes a restriction on the shape of the density function $f$ that can be utilized for estimation. To facilitate estimation, we consider an equivalent condition of Myerson regularity under a mild technical assumption.

\begin{assumption} \label{ass:f-continuity}
    The density function $f$ is continuous and satisfies the following condition:
    \begin{align*}
        \limsup_{\varepsilon \rightarrow 0^+} (f(v+\varepsilon) - f(v))/\varepsilon > -\infty, \text{ for every } v \in [\underline{v},\bar{v}] \text{ except possibly at a countable set.}
    \end{align*}
    
\end{assumption}

Besides the continuity of $f$, Assumption \ref{ass:f-continuity} also requires that the upper Dini derivative $\limsup_{\varepsilon \rightarrow 0^+} (f(v+\varepsilon) - f(v))/\varepsilon$ is finite. This is a very mild technical condition that can be satisfied if, for example, the density $f$ is locally Lipschitz continuous. Following \cite{ewerhart2013regular}, we can now equivalently state Myerson regularity as follows.

\begin{proposition} \label{lm:equiv-myerson-regular}
    Under Assumption \ref{ass:f-continuity}, Myerson regularity (Assumption \ref{ass:myerson-regularity}) is equivalent to the following condition:
    \begin{align} \label{eqn:convex-1/(1-F)}
        \Lambda(v) \equiv (1-F(v))^{-1} \text{ is a convex function on } [\underline{v},\bar{v}].
    \end{align}
\end{proposition}

We can better illustrate Proposition \ref{lm:equiv-myerson-regular} by assuming that the density $f$ is differentiable with derivative $f'$. Denote $\lambda$ as the derivative of $\Lambda$:
\begin{align*}
    \lambda(v) \equiv \Lambda'(v) = f(v)/(1-F(v))^2.
\end{align*}
Then we have 
\begin{align*}
    \varphi'(v) & = \frac{2f(v)^2 + (1-F(v))f'(v)}{f(v)^2}, \\
    \lambda'(v) & = \frac{2f(v)^2 + (1-F(v))f'(v)}{(1-F(v))^3}.
\end{align*}
The respective signs of $\varphi'(v)$ and $\lambda'(v)$ coincide and are both determined by the sign of $2f(v)^2 + (1-F(v))f'(v)$. Proposition \ref{lm:equiv-myerson-regular} is the generalization of this equivalence result to the case where $f$ is not differentiable.\footnote{Condition (\ref{eqn:convex-1/(1-F)}) was first pointed out by \cite{mcafee1987auctions} in footnote 11. \cite{ewerhart2013regular} drew a connection between $f/(1-F)^{-2}$, the derivative of $(1-F)^{-1}$, and the probability rate of being the next to fail. \cite{SZECH2011optimal} offered a different perspective on the condition, relating it to the monotonicity of the sequence of increments of expected second order statistics. Lastly, \cite{fang2017nonparametric} provided an additional interpretation of the condition, viewing it as the convexity of odds ratio.} 

Condition (\ref{eqn:convex-1/(1-F)}) is the key to the estimation procedure we propose. Based on the empirical distribution function $F_n(v) \equiv \frac{1}{n} \sum_{i=1}^n \mathbf{1}\{V_i \leq v\}$, we can construct an estimator for $\Lambda(v)$ as 
\begin{align*}
    \Lambda_n(v) \equiv (1 - F_n(n) + 1/n)^{-1},
\end{align*}
where we add the term $1/n$ in the denominator to avoid dividing by zero. The estimator $\Lambda_n$ is converging to $\Lambda$ but may not be convex in finite samples. Let $\hat{\Lambda}_n$ be the greatest convex minorant (gcm) of $\Lambda_n$. That is, $\hat{\Lambda}_n$ is the largest convex function that lies below $\Lambda_n$. Since $\hat{\Lambda}_n$ is convex, it is almost everywhere differentiable. Take $\hat{\lambda}_n$ as the left-derivative of $\hat{\Lambda}_n$. Since $\hat{\lambda}_n$ is an estimator for $\lambda = f/(1-F)^2$, we can naturally estimate $f$ by the estimator 
\begin{align*}
    \hat{f}_n(v) \equiv \hat{\lambda}_n(v) (1-F_n(v))^2.
\end{align*}
This nonparametric estimator is tuning-parameter-free. In particular, we do not need to choose the bandwidth as in kernel-based estimators.

\section{Asymptotic properties}

In this section, we derive the asymptotic properties of the estimator $\hat{f}_n$. It is well-known that Grenander-type estimators can behave badly near the boundary points \citep{woodroofe1993penalized,kulikov2006behavior,balabdaoui2011grenander}. Therefore, we focus on a subinterval $[a,b]$ of the support $[\underline{v},\bar{v}]$, where $\underline{v} < a < b < \bar{v}$.

\subsection{Consistency and asymptotic distribution}

We first derive the uniform consistency of the estimator $\Lambda_n$, which can be derived based on the uniform consistency of the empirical distribution. 

\begin{lemma} \label{lm:consistency-Lambda}
    $\Lambda_n$ is uniformly consistent over $[a,b]$, that is,
    \begin{align*}
        \sup_{v \in [a,b]}|\hat{\Lambda}_n(v) - \Lambda(v)| = o_p(1).
    \end{align*}
\end{lemma}

The following theorem states the (uniform) consistency of $\hat{f}_n(v)$ under the (uniform) continuity of the density function $f$. The idea is that (uniform) conintuity corresponds to (uniform) consistency.

\begin{theorem}[Consistency] \label{thm:consistency}
    Let Assumptions \ref{ass:myerson-regularity} - \ref{ass:f-continuity} hold. 
    \begin{enumerate} [label = (\roman*)]
        \item For any $v \in [a,b]$, $\hat{f}_n(v) = f(v) + o_p(1)$.
        \item If we further assume that the density $f$ is uniformly continuous on $[\underline{v},\bar{v}]$, then 
        \begin{align*}
        \sup_{v \in [a,b]}|\hat{f}_n(v) - f(v)| = o_p(1).
    \end{align*}
    \end{enumerate}
    
\end{theorem}

Because $\hat{\lambda}_n$ is the derivative of $\hat{\Lambda}_n$, to derive its asymptotic distribution, we need to first study the local behavior of $\Lambda_n - \Lambda$. Define the stochastic process $J_n(t)$ as 
\begin{align*}
    J_n(t) \equiv n^{2/3}\left(\Lambda_n(v + tn^{-1/3})  - \Lambda(v + tn^{-1/3}) - (\Lambda_n(v) - \Lambda(v)) \right).
\end{align*}
The following lemma establishes the weak convergence of $J_n(t)$, which is useful in deriving the asymptotic distribution of $\hat{f}_n(v)$.

\begin{lemma} \label{lm:asymp-dist-Lambda}
Let Assumptions \ref{ass:myerson-regularity} - \ref{ass:f-continuity} hold. For any $v \in [a,b]$, we have
    \begin{align*}
        J_n(t) \overset{d}{\rightarrow} \frac{\sqrt{f(v)}}{(1-F(v))^2} \mathbf{B}(t),
    \end{align*}
    where $\mathbf{B}(t)$ is a two-sided Brownian motion.
\end{lemma}

\begin{theorem}[Asymptotic distribution] \label{thm:asymp-dist}
    Let Assumptions \ref{ass:myerson-regularity} - \ref{ass:f-continuity} hold. If we further assume that the density $f$ is continuously differentiable on $[\underline{v},\bar{v}]$, then for any $v \in [a,b]$,
    \begin{align*}
        n^{1/3}(\hat{f}_n(v) - f(v)) \overset{d}{\rightarrow} C(v) Z,
    \end{align*}
    where $Z\equiv \argmax_{t \in \mathbb{R}} \{ \mathbf{B}(t) - t^2 \}$, and the constant $C(v)$ is 
    \begin{align*}
        C(v) \equiv \left( \frac{8f(v)^3}{1-F(v)} + 4f(v)f'(v) \right)^{1/3}.
    \end{align*}
\end{theorem}

The distribution of $Z$, the argmax of two-sided brownian motion with quadratic drift, is referred to as Chernoff's distribution as it first arose in \cite{chernoff1964estimation} on mode estimation. The density, distribution function, quantiles, and moments of Chernoff's distribution are computed in \cite{groeneboom2001computing}. In particular, its density function is symmetric around zero, and hence our estimator does not have asymptotic bias. In contrast, kernel density estimators often has asymptotic bias when converging at the minimax optimal rate. The variance of $Z$ is approximately $0.26$. As suggested by \cite{dykstra1999distribution}, the distribution of $Z$ can be approximated by the normal distribution $N(0,(0.52)^2)$.

To conduct inference on $f(v)$, one can estimate the density derivative $f'(v)$ using the conventional methods and obtain an estimate for $C(v)$. To obtain the quantiles of $Z$, one can use Table 3 in \cite{groeneboom2001computing}. Notice that it is difficult to obtain a cube-root test statistic based on the kernel density estimator when the density is only first-order differentiable. In general, inference procedures need undersmoothing to eliminate the asymptotic bias, thus leading to suboptimal convergence rates.

\subsection{Minimax optimality}

We are interested to know whether our estimator attains the optimal rate of convergence given the current set of assumptions that we are considering. To do that, our goal is to derive lower bounds on the convergence rate achievable by any estimation procedure. These lower bounds represent the intrinsic difficulty of the estimation problem at hand. From \cite{stone1980optimal}, we know that the lower bound for estimating a continuously differentiable density is $n^{-1/3}$, which is achievable by the kernel density estimator. Nonetheless, our problem deviates from this scenario due to the incorporation of Myerson regularity in addition to smoothness.

Notice that evaluating the estimators' performance with respect to a particular density is not feasible. This is due to the existence of an invariably superior estimation method: simply discard the data and return that particular density function. Consequently, we should focus on assessing the performance of the estimators across a set of distributions, in a minimax sense. Therefore, we need to consider the performance of the set of estimators over a set of distributions in the minimax sense. In our case, the relevant set of densities is $\mathcal{F}$, defined as the set of distributions that have a.e. continuously differentiable densities and satisfy Assumptions \ref{ass:myerson-regularity} - \ref{ass:f-continuity}. The following theorem demonstrates a minimax lower bound on the convergence rate of any estimator as $n^{-1/3}$ (multiplied by a constant). Therefore, our estimator $\hat{f}_n$ achieves the optimal rate.

\begin{theorem}[Minimax optimal convergence rate] \label{thm:minimax}
    For any $v \in [a,b]$, there exists $c>0$ such that
    \begin{align*}
        \inf_{\tilde{f}_n} \sup_{f \in \mathcal{F}} \mathbb{E}_f|\tilde{f}_n(v) - f(v)| \geq c n^{-1/3}, 
    \end{align*}
    where $\mathbb{E}_f$ denotes the expectation with respect to the distribution $f$. The infimum $\inf_{\tilde{f}_n}$ is taken over the set of all estimators.
\end{theorem}

\section{Conclusion}

This paper applies the Myerson regularity condition as a shape constraint on the valuation distribution for nonparametric density estimation. We introduce a nonparametric estimator that is entirely data-driven and does not require tuning parameters. We demonstrate the consistency of this estimator at the cube-root rate, which is proven to be the minimax optimal convergence rate. We further derive the asymptotic Chernoff's distribution of the estimator and describe valid inference procedures.

\appendix 

\numberwithin{equation}{section}
\numberwithin{lemma}{section}
\numberwithin{definition}{section}

\section{Technical proofs}

\begin{proof}[Proof of Proposition \ref{lm:equiv-myerson-regular}]
    Since the density $f$ is continuous, it is also right-continuous and upper semi-continuous. Then the result follows from Lemma 4.1 in \cite{ewerhart2013regular}.
\end{proof}

\begin{proof}[Proof of Lemma \ref{lm:consistency-Lambda}]
    By the mean value theorem, there exists $\tilde{\xi}$ between $F_n(v)$ and $F(v)$ such that 
    \begin{align} \label{eqn:mean-value-Lambda}
        \Lambda_n(v) - \Lambda(v) = (1-\tilde{\xi})^{-2} (F_n(v) - F(v) + 1/n), v \in [a,b].
    \end{align}
    Since $F_n$ and $F$ are nondecreasing, and the function $(1-\cdot)^{-2}$ is strictly increasing , $(1-\tilde{\xi})^{-2}$ is bounded by $(1- F_n(b) \vee F(b) )^{-2}$ for $v \in [a,b]$.
    Taking the supremum over $[a,b]$, we have 
    \begin{align*}
        \sup_{v \in [a,b]} |\Lambda_n(v) - \Lambda(v)| & = \sup_{v \in [a,b]} |(1-\tilde{\xi})^2 (F_n(v) - F(v))| \\
        & \leq (1- F_n(b) \vee F(b) )^{-2} \Big( \sup_{v \in [a,b]} | F_n(v) - F(v) | + 1/n \Big).
    \end{align*}
    By the Glivenko-Cantelli theorem \citep[see, e.g., Theorem 19.1 in][]{vaart_1998}, we know that $\sup_{v \in [a,b]} | F_n(v) - F(v) | = o_p(1)$. The remaining task is to show that the term $(1- F_n(b) \vee F(b) )^{-2}$ is bounded in probability. Since $F$ is continuous, there exists $\tilde{b} \in (b,\bar{v})$ such that $F(b) < F(\tilde{b}) < F(\bar{v})=1$. By the strict monotonicity of the function $(1-\cdot)^{-2}$, we have 
    \begin{align*}
        \mathbb{P}\left( (1- F_n(b) \vee F(b) )^{-2} > (1-F(\tilde{b}))^{-2} \right) & = \mathbb{P}\left( F_n(b) > F(\tilde{b}) \right) \\
        & \leq \mathbb{P}\left( |F_n(b) - F(b)| > F(\tilde{b}) - F(b) \right) \rightarrow 0.
    \end{align*}
    This shows that $(1- F_n(b) \vee F(b) )^{-2} = O_p(1)$.
\end{proof}

\begin{proof}[Proof of Theorem \ref{thm:consistency}]
    We invoke Theorem 1 in \cite{westling2020unified} to prove the consistency of $\hat{\lambda}_n$. To comply with the notations in that paper, we can define $\theta_0 = \lambda$, $\theta_n = \hat{\lambda}_n$, $I = J_0 = [a,b]$, $\Psi_0 = \Gamma_0 = \Lambda$, $\Psi_n = \Gamma_n = \Lambda_n$ and $\Phi_0 = \Phi_n = \text{id}$, where $\text{id}$ is the identity mapping. By construction, $\Phi_0$ is strictly increasing and uniformly continuous. In view of Theorem 1 in \cite{westling2020unified}, we only need to show the (uniform) continuity of $\lambda$.
    
    For part (i) of the theorem, we know that $f$ is continuous, and $(1-F)^2$ is continuous and bounded away from zero on $[a,b]$. This implies that $\lambda$ is continuous on $[a,b]$. By Theorem 1 in \cite{westling2020unified}, we know that $\hat{\lambda}_n(v) = \lambda(v) + o_p(1)$. Then we know that the estimator $\hat{f}_n(v) = \hat{\lambda}_n(v) (1-F_n(v))^2$ is pointwise consistent for $f(v)$ by an application of Slutsky's theorem. 
    
    For part (ii) of the theorem, we know that $f$ is uniformly continuous, and $(1-F)^2$ is uniformly continuous and bounded away from zero on $[a,b]$. This implies that $\lambda$ is uniformly continuous on $[a,b]$. By Theorem 1 in \cite{westling2020unified}, we know that $\sup_{v \in [a,b]}|\hat{\lambda}_n(v) - \lambda(v)| = o_p(1)$. By the triangle inequality, we have 
    \begin{align*}
        \sup_{v \in [a,b]}|\hat{f}_n(v) - f(v)| & = \sup_{v \in [a,b]}|\hat{\lambda}_n(v)(1-F_n(v))^2 - \lambda(v)(1-F(v))^2| \\
        & \leq \sup_{v \in [a,b]}|(\hat{\lambda}_n(v) - \lambda(v))(1-F_n(v))^2| \\
        & + \sup_{v \in [a,b]}|\lambda(v)((1-F_n(v))^2 - (1-F(v))^2)|.
    \end{align*}
    The term $\sup_{v \in [a,b]}|(\hat{\lambda}_n(v) - \lambda(v))(1-F_n(v))^2|$ is $o_p(1)$ since $(1-F_n(v))^2$ is bounded. The term $\sup_{v \in [a,b]}|\lambda(v)((1-F_n(v))^2 - (1-F(v))^2)|$ is $o_p(1)$ since $\lambda(v)$ is bounded (due to continuity), and $((1-F_n(v))^2$ is uniformly consistent, which can be shown in a way analogous to Lemma \ref{lm:consistency-Lambda}. 
\end{proof}

\begin{lemma} \label{lm:asymp-dist-F}
For any $v \in [a,b]$,
    \begin{align*}
        n^{2/3} \left(F_n(v + tn^{-1/3})  - F(v + tn^{-1/3}) - (F_n(v) - F(v)) \right) \overset{d}{\rightarrow} \sqrt{f(v)} \mathbf{B}(t),
    \end{align*}
    where $\mathbf{B}(t)$ is a two-sided Brownian motion.
\end{lemma}

\begin{proof}[Proof of Lemma \ref{lm:asymp-dist-F}]
    We start with the case $t \in [0,K]$, where $K$ is an arbitrary positive integer. Define
    \begin{align*}
        Z_{ni}(t) \equiv n^{-1/3}\left( \mathbf{1}\{v < V_i \leq v + tn^{-1/3}\} - (F(v+tn^{-1/3})-F(v)) \right), t \geq 0.
    \end{align*}
    We use Theorem 2.11.1 in \cite{wellner1996} to prove the weak convergence of $\sum_{i=1}^n Z_{ni}(t)$. There are three conditions in that theorem to be verified. The first Lindeberg condition can be verified similar to Theorem 6 in \cite{durot2007mathbb}. The third entropy condition, (2.11.2), is verified by Condition (2.5.1) in that book because the index $t \in [0,K]$ is one-dimensional. For the second condition, we have for any $s<t$,
    \begin{align*}
        & \mathbb{E}[(Z_{ni}(t) - Z_{ni}(s))^2] \\
        = & n^{-2/3} \mathbb{E}\left[ \left( \mathbf{1}\{v+sn^{-1/3} < V_i \leq v + tn^{-1/3}\} - (F(v+tn^{-1/3})-F(v+sn^{-1/3})) \right)^2 \right] \\
        = & n^{-2/3} (F(v+tn^{-1/3})-F(v+sn^{-1/3})) - n^{-2/3}(F(v+tn^{-1/3})-F(v+sn^{-1/3}))^2.
    \end{align*}
    By the mean value theorem and the continuity of $f$, we have 
    \begin{align*}
        \sup_{|t-s|<\delta}|F(v+tn^{-1/3})-F(v+sn^{-1/3})| \leq \lvert f \rvert_\infty \delta n^{-1/3}.
    \end{align*}
    Therefore, $\sup_{|t-s|<\delta} \sum_{i=1}^n \mathbb{E}[(Z_{ni}(t) - Z_{ni}(s))^2] = O(\delta)$, which proves the second condition of that theorem. 
    
    Then we need to check the pointwise limit of the covariance function. The covariance is equal to 
    \begin{align*}
        \mathbb{E}\left[ \sum_{i,j=1}^n Z_{ni}(t) Z_{nj}(s) \right] = \mathbb{E}\left[ \sum_{i=1}^n Z_{ni}(t) Z_{ni}(s) \right] = n \mathbb{E}\left[ Z_{ni}(t) Z_{ni}(s) \right].
    \end{align*}
    By the definition of $Z_{ni}$, we have 
    \begin{align*}
        \mathbb{E}[ Z_{ni}(t) Z_{ni}(s) ] & = n^{-2/3} (F(v+(s\wedge t)n^{-1/3})-F(v)) \\
        & - n^{-2/3} (F(v+tn^{-1/3})-F(v+sn^{-1/3}))^2.
    \end{align*}
    Since $F$ is continuously differentiable with derivative $f$, we have 
    \begin{align*}
        F(v+(s \wedge t)n^{-1/3})-F(v) & = f(v) (s \wedge t)n^{-1/3} + o(n^{-1/3}), \\
        (F(v+tn^{-1/3})-F(v+sn^{-1/3}))^2 & = O(n^{-2/3}) = o(n^{-1/3}).
    \end{align*}
    Then $n \mathbb{E}\left[ Z_{ni}(t) Z_{ni}(s) \right] = f(v)(s \wedge t) + o(1)$. Therefore, we know that, on the interval $[0,K]$, $\sum_{i=1}^n Z_{ni}(t)$ weakly converges to a mean zero Gaussian process with covariance function $f(v)(s \wedge t)$. This limit process can be presented as $\sqrt{f(v)} \mathbf{B}(t)$. Since $K$ is arbitrary, the weak convergence can be established on the positive real line $\mathbb{R} = \bigcup_{K=1}^\infty [0,K]$ in view of Theorem 1.6.1 in \cite{wellner1996}. Lastly, we can extend the weak convergence to the entire real line by redefining $Z_{in}(t)$ as 
    \begin{align*}
        Z_{ni}(t) \equiv n^{-1/3}\left( \mathbf{1}\{v + tn^{-1/3} < V_i \leq v \} - (F(v+tn^{-1/3})-F(v)) \right), t \leq 0.
    \end{align*}
    The remaining arguments hold analogously.
\end{proof}

\begin{proof}[Proof of Lemma \ref{lm:asymp-dist-Lambda}]
    By the mean value theorem expressed in (\ref{eqn:mean-value-Lambda}), we have 
    \begin{align*}
        J_n(t) = & n^{2/3} (1-\xi_{1n}(t))^{-2} (F_n(v+tn^{-1/3}) + 1/n - F(v+tn^{-1/3})) \\
        & - n^{2/3} (1-\xi_{2n})^{-2} (F_n(v) + 1/n - F(v)),
    \end{align*}
    where $\xi_{1n}(t)$ is between $F_n(v+tn^{-1/3})$ and $F(v+tn^{-1/3})$, and $\xi_{2n}$ is between $F_n(v)$ and $F(v)$. We decompose $J_n(t)$ into three terms:
    \begin{align*}
        J_n(t) = \widetilde{J}_n(t) + \text{err}_{1n}(t) + \text{err}_{2n},
    \end{align*}
    where we define 
    \begin{align*}
        \widetilde{J}_n(t) \equiv & n^{2/3} (1-F(v))^{-2} \left(F_n(v + tn^{-1/3})  - F(v + tn^{-1/3}) - (F_n(v) - F(v)) \right), \\
        \text{err}_{1n}(t) \equiv & n^{2/3} ((1-\xi_{1n}(t))^{-2} - (1-F(v))^{-2}) (F_n(v+tn^{-1/3}) - F(v+tn^{-1/3})) \\
        & + n^{2/3} (1-\xi_{1n}(t))^{-2}/n, \\
        \text{err}_{2n} \equiv & n^{2/3} ((1-\xi_{2n})^{-2} - (1-F(v))^{-2}) (F_n(v) - F(v)) + n^{2/3} (1-\xi_{2n})^{-2}/n.
    \end{align*}
    By Lemma \ref{lm:asymp-dist-F}, we know that $\widetilde{J}_n(t)$ converges weakly to $\frac{\sqrt{f(v)}}{(1-F(v))^2} \mathbf{B}(t)$. The remaining task is to show that the two error terms are negligible.
    
    The supremum of the first error term, $\sup_{t \in [0,K]} |\text{err}_{1n}(t)|$, is bounded by 
    \begin{align*}
        & n^{2/3} \sup_{t \in [0,K]} |(1-\xi_{1n}(t))^{-2} - (1-F(v))^{-2}| \sup_{t \in [0,K]}|F_n(v+tn^{-1/3}) - F(v+tn^{-1/3})| \\
        + & n^{-1/3} \sup_{t \in [0,K]} |(1-\xi_{1n}(t))^{-2}|.
    \end{align*}
    Based on the uniform convergence rate of the empirical cumulative distribution function \citep[Dvoretzky–Kiefer–Wolfowitz inequality, p. 268,][]{vaart_1998}, we have 
    \begin{align*}
        \sup_{t \in [0,K]}|F_n(v+tn^{-1/3}) - F(v+tn^{-1/3})| \leq \sup_{v \in \mathbb{R}}|F_n(v) - F(v)| = O_p(n^{-1/2}).
    \end{align*}
    Because the function $(1-\cdot)^{-2}$ is monotonic, we know that, for any $t \in [0,K]$,
    \begin{align*}
        & |(1-\xi_{1n}(t))^{-2} - (1-F(v))^{-2}| \\
        \leq &  |(1-F(v+tn^{-1/3}))^{-2} - (1-F(v))^{-2}| \vee |(1-F_n(v+tn^{-1/3}))^{-2} - (1-F(v))^{-2}| \\
        \leq & |(1-F(v+tn^{-1/3}))^{-2} - (1-F(v))^{-2}| + |(1-F_n(v+tn^{-1/3}))^{-2} - (1-F(v))^{-2}| \\
        \leq & 2|(1-F(v+tn^{-1/3}))^{-2} - (1-F(v))^{-2}| \\
        & + |(1-F_n(v+tn^{-1/3}))^{-2} - (1-F(v+tn^{-1/3}))^{-2}|,
    \end{align*}
    where the triangle inequality is used in the last step.
    For the first term on the RHS, as $F$ is non-decreasing and $(1-\cdot)^{-3}$ is increasing, applying Taylor expansion twice give that
    \begin{align*}
        & \sup_{t \in [0,K]} |(1-F(v+tn^{-1/3}))^{-2} - (1-F(v))^{-2}| \\
        \leq & 2 \sup_{t \in [0,K]}(1-F(v+tn^{-1/3}))^{-3} |F(v+tn^{-1/3}) - F(v)| \\
        \leq & 2 (1-F(v+Kn^{-1/3}))^{-3} \sup_{v \in \mathbb{R}} |f(v)| Kn^{-1/3} = O(n^{-1/3}),
    \end{align*}
    where we have utilized the condition that $f$ is bounded (since it is continuous on the compact support) and the fact that the term $F(v+Kn^{-1/3})$ is strictly less than $1$ when $n$ is sufficiently large. For the second term, following the same reasoning as in the proof of Lemma \ref{lm:consistency-Lambda}, we know it is bounded by 
    \begin{align*}
        O_p(1)\times \sup_{t \in [0,K]}|F_n(v+tn^{-1/3}) - F(v+tn^{-1/3})| = O_p(n^{-1/2}).
    \end{align*}
    Therefore, we have 
    \begin{align*}
        &n^{2/3} \sup_{t \in [0,K]} |(1-\xi_{1n}(t))^{-2} - (1-F(v))^{-2}| \sup_{t \in [0,K]}|F_n(v+tn^{-1/3}) - F(v+tn^{-1/3})| \\
        =& n^{2/3} O_p(n^{-5/6}) = o_p(1).
    \end{align*}
    As a byproduct of the above analysis, we have also shown that the term $(1-\xi_{1n}(t))^{-2}$ is bounded in probability. Hence, we have $\sup_{t \in [0,K]} |\text{err}_{1n}(t)| = o_p(1)$. The second error term does not depend on $t$ can be shown to be $o_p(1)$ following the same reasoning.
    
\end{proof}

\begin{proof}[Proof of Theorem \ref{thm:asymp-dist}]
    We follow the proof of Theorem 1 in \cite{luo2018integrated} to derive the asymptotic distribution of $\hat{\lambda}_n$. Let 
    \begin{align*}
        U_n(c) = \argmin_{s \in [a,b]} \{\Lambda_n(s) - cs \}, c>0.
    \end{align*}
    The process $U_n(a)$, first proposes by \cite{groeneboom1983concave}, is a very useful tool in deriving the asymptotic distribution $\hat{\lambda}_n$. For any $v \in [a,b]$, we have the switching relation: $U_n(c) \geq v \iff \hat{\lambda}_n \leq c$. Then by the definition of $U_n$, we have 
    \begin{align*}
        n^{1/3}(\hat{\lambda}_n(v) - \lambda(v)) \leq z & \iff \hat{\lambda}_n(v) \leq \lambda(v) + zn^{-1/3} \\
        & \iff U_n(\lambda(v) + zn^{-1/3}) \geq v \\
        & \iff \argmin_{s \in [a,b]} \{\Lambda_n(s) - (\lambda(v) + zn^{-1/3})s \} \geq v
    \end{align*}
    After changing variable $s = v + tn^{-1/3}$, the above event regarding the minimization can be equivalently written as 
    \begin{align*}
        & \argmin_{s \in [a,b]} \{\Lambda_n(s) - (\lambda(v) + zn^{-1/3})s \} \geq v \\
        \iff & \argmin_{t \in n^{1/3}[a-v,b-v]} \{\Lambda_n(v + tn^{-1/3}) - (\lambda(v) + zn^{-1/3})(v + tn^{-1/3}) \} \geq 0
    \end{align*}
    Since the argmin does not change when we add or multiply constants to the objective function, we have 
    \begin{align*}
        & \argmin_{t \in n^{1/3}[a-v,b-v]} \{\Lambda_n(v + tn^{-1/3}) - (\lambda(v) + zn^{-1/3})(v + tn^{-1/3}) \} \\
        = & \argmin_{t \in n^{1/3}[a-v,b-v]} \{n^{2/3}\Lambda_n(v + tn^{-1/3}) - n^{2/3}\Lambda_n(v) - n^{1/3}\lambda(v) t - zt \},
    \end{align*}
    Define the process $W_n(t)$ by 
    \begin{align*}
        W_n(t) \equiv n^{2/3}\Lambda_n(v + tn^{-1/3}) - n^{2/3}\Lambda_n(v) - n^{1/3}\lambda(v) t.
    \end{align*}
    The above analysis shows that 
    \begin{align} \label{eqn:switching}
        n^{1/3}(\hat{\lambda}_n(v) - \lambda(v)) \leq z \iff \argmin_{t \in n^{1/3}[a-v,b-v]} \{W_n(t) - zt\} \geq 0.
    \end{align}
    To study its asymptotic behavior, we can decompose $W_n(t)$ as 
    \begin{align*}
        W_n(t) & = n^{2/3}(\Lambda_n(v + tn^{-1/3}) - \Lambda_n(v) - (\Lambda(v + tn^{-1/3}) - \Lambda(v))) \\
        & + n^{2/3}(\Lambda(v + tn^{-1/3}) - \Lambda(v) - n^{-1/3}\lambda(v) t).
    \end{align*}
    By Lemma \ref{lm:asymp-dist-Lambda}, we know that the first term weakly converges to $\frac{\sqrt{f(v)}}{(1-F(v))^2} \mathbf{B}(t)$. For the second term, notice that $\lambda = f/(1-F)^2$ is continuously differentiable over $[a,b]$ because $f$ is assume to be continuously differentiable. Then for any $K>0$, we have uniformly over $t \in [-K,K]$, 
    \begin{align*}
        n^{2/3}(\Lambda(v + tn^{-1/3}) - \Lambda(v) - n^{-1/3}\lambda(v) t) = \lambda'(v) t^2/2 + o(1).
    \end{align*}
    In view of Theorem 1.6.1 in \cite{wellner1996}, we have 
    \begin{align*}
        W_n(t) \overset{d}{\rightarrow} \frac{\sqrt{f(v)}}{(1-F(v))^2} \mathbf{B}(t) + \frac{\lambda'(v)}{2} t^2.
    \end{align*}
    By the Argmax Theorem, that is, Theorem 3.2.2 in \cite{wellner1996}, we know that 
    \begin{align*}
        \argmin_{t \in n^{1/3}[a-v,b-v]} \{W_n(t) - zt\} \overset{d}{\rightarrow} \argmin_{t \in \mathbb{R}} \left\{ \alpha \mathbf{B}(t) + \beta t^2 -zt \right\},
    \end{align*}
    where, for simplicity, we denote $\alpha \equiv \frac{\sqrt{f(v)}}{(1-F(v))^2}$ and $\beta \equiv \frac{\lambda'(v)}{2}$. Following the proof of Theorem 1 in \cite{luo2018integrated}, we can show that 
    \begin{align*}
        \argmin_{t \in \mathbb{R}} \left\{ \alpha \mathbf{B}(t) + \beta t^2 -zt \right\} \overset{d}{\sim} \left( \frac{\alpha}{\beta} \right)^{2/3} \argmin_{t \in \mathbb{R}} \left\{ \mathbf{B}(t) + t^2 \right\} + \frac{z}{2\beta}.
    \end{align*}
    From the relationship in (\ref{eqn:switching}), we have 
    \begin{align*}
        \mathbb{P}\left( n^{1/3}(\hat{\lambda}_n(v) - \lambda(v)) \leq z \right) & \rightarrow \mathbb{P} \left( \left( \frac{\alpha}{\beta} \right)^{2/3} \argmin_{t \in \mathbb{R}} \left\{ \mathbf{B}(t) + t^2 \right\} + \frac{z}{2\beta} \geq 0 \right) \\
        & = \mathbb{P} \left( 2\alpha^{2/3}\beta^{1/3} \argmax_{t \in \mathbb{R}} \{ \mathbf{B}(t) - t^2 \} \leq z\right).
    \end{align*}
    Therefore, we have the asymptotic distribution of $\hat{\lambda}_n(v)$ as the following:
    \begin{align*}
        n^{1/3}(\hat{\lambda}_n(v) - \lambda(v)) \overset{d}{\rightarrow} 2\alpha^{2/3}\beta^{1/3} \argmax_{t \in \mathbb{R}} \{ \mathbf{B}(t) - t^2 \}.
    \end{align*}
    Lastly, to derive the asymptotic distribution of $\hat{f}_n(v)$, we have 
    \begin{align*}
        n^{1/3}(\hat{f}_n(v) - f(v)) & = n^{1/3}(\hat{\lambda}_n(v)(1-F_n(v))^2 - \lambda(v)(1-F(v))^2) \\
        & = n^{1/3}(\hat{\lambda}_n(v)(1-F_n(v))^2 - \hat{\lambda}_n(v)(1-F(v))^2) \\
        & + n^{1/3}(\hat{\lambda}_n(v)(1-F(v))^2 - \lambda(v)(1-F(v))^2).
    \end{align*} 
    The first term is $o_p(1)$ because $|\hat{\lambda}_n(v)| = O_p(1)$, and 
    \begin{align*}
        |(1-F_n(v))^2 - (1-F(v))^2| & = |F_n(v) - F(v)| |2 - F_n(v) - F(v)| \\
        & \leq 2|F_n(v) - F(v)| = O_p(n^{-1/2}).
    \end{align*}
    The second term converges in distribution to $2\alpha^{2/3}\beta^{1/3}(1-F(v))^2 \argmax_{t \in \mathbb{R}} \{ \mathbf{B}(t) - t^2 \}$. This proves the result that 
    \begin{align*}
        n^{1/3}(\hat{f}_n(v) - f(v)) \overset{d}{\rightarrow} C(v) \argmax_{t \in \mathbb{R}} \{ \mathbf{B}(t) - t^2 \}.
    \end{align*}

\end{proof}

\begin{proof}[Proof of Theorem \ref{thm:minimax}]
    For simplicity in notation, we choose the support to be $[\underline{v},\bar{v}]=[0,1]$ and $v = 0.5$. The other cases can be analyzed analogously. We use the method described in Chapter 15.2 of \cite{wainwright2019high} to prove the minimax lower bound on the convergence rate. Consider $f(v)$, the evaluation of the density function $f$ at the point $v \in [a,b]$, as an evaluation functional on the set of densities $\mathcal{F}$ defined in Theorem \ref{thm:minimax}. Define $\omega(\epsilon)$ as the modulus of continuity of the evaluation functional (at $v=0.5$) with respect to the Hellinger norm on $\mathcal{F}$, that is,
    \begin{align*}
        \omega(\epsilon) \equiv \sup_{f_1,f_2 \in \mathcal{F}} \{|f_1(0.5) - f_2(0.5)|: H(f_1 \| f_2) \leq \epsilon\},
    \end{align*}
    where the Hellinger norm $H(f_1 \| f_2)$ is given by
    \begin{align*}
        H(f_1 \| f_2)^2 \equiv \int \left(\sqrt{f_1(v)} - \sqrt{f_2(v)} \right)^2 dv.
    \end{align*}
    By Corollary 15.6 (Le Cam for functionals) in Chapter 15.2 of \cite{wainwright2019high}, we know that the minimax risk is lower bounded as
    \begin{align*}
        \inf_{\tilde{f}} \sup_{f \in \mathcal{F}} \mathbb{E}_f |\tilde{f}(v) - f(v)| \geq \frac{1}{8} \omega\left(1/(2\sqrt{n})\right).
    \end{align*}
    Our remaining task is to characterize the modulus of continuity $\omega$. We want to find two density functions in $\mathcal{F}$ that are close when measured with the Hellinger distance, but their evaluation functionals are well-separated. We define the following two densities:
    \begin{align}
        f_1(v) & \equiv 1, v \in [\underline{v},\bar{v}] \nonumber \\
        f_2(v) & \equiv 1 + \delta \phi((v-0.5)/\delta), v \in [\underline{v},\bar{v}] \label{eq:f_2}.
    \end{align}
    The density $f_1$ is the uniform distribution. We add a small perturbation to the uniform density to obtain the density $f_2$. The coefficient $\delta$ depends on $n$ and is specified later in the proof. The perturbation function $\phi$ is defined as 
    \begin{align*}
        \phi(t) \equiv 
        \begin{cases}
        t+1, & t\in[-1,0],\\
        -t+1, & t\in[0,2],\\
        t-3, & t\in[2,3],\\
        0, & \text{ otherwise.}
        \end{cases}
    \end{align*}
    We graph the perturbation function $\phi$ and the perturbed density $f_2$ below.
    
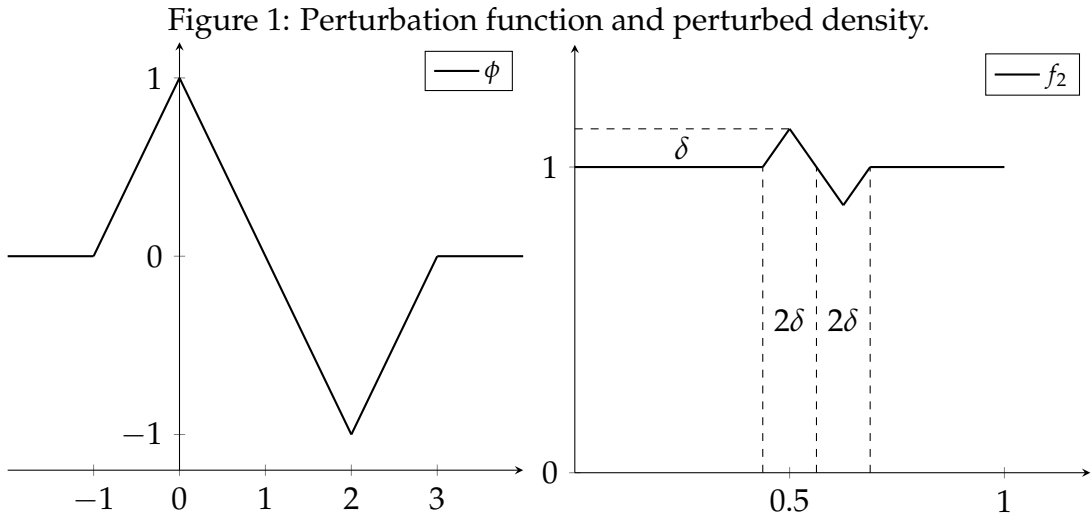
\begin{figure}[!htbp]
\caption{Perturbation function and perturbed density.}
\centering \begin{tikzpicture}
				
				\begin{axis}[
					legend style={nodes={scale=0.8, transform shape}},
					axis y line=center, 
					axis x line=bottom,
					%y axis line style={opacity=0},
					ytick={-1,0,1},
					xtick={-1,0,1,2,3},
					xmin=-2,     xmax=4,
					ymin=-1.2,     ymax=1.2,
					%xlabel = \(u\),
					%ylabel = {\(MTE\)},
					]
					%Below the red parabola is defined
					\addplot [
					thick,
					domain=-2:-1, 
					samples=100, 
					color=black,
					]
					{0};
					\addplot [
					thick,
					domain=-1:0, 
					samples=100, 
					color=black,
					]
					{x+1};
					\addplot [
					thick,
					domain=0:2, 
					samples=100, 
					color=black,
					]
					{-x+1};
					\addplot [
					thick,
					domain=2:3, 
					samples=100, 
					color=black,
					]
					{x-3};
					\addplot [
					thick,
					domain=3:4, 
					samples=100, 
					color=black,
					]
					{0};
					\addlegendentry{$\phi$};
					
				\end{axis}

			\end{tikzpicture} \begin{tikzpicture}
				
				\begin{axis}[
					legend style={nodes={scale=0.8, transform shape}},
					axis y line=left, 
					axis x line=bottom,
					%y axis line style={opacity=0},
					ytick={0,1},
					xtick={0.5,1},
					xmin=0,     xmax=1.2,
					ymin=0,     ymax=1.4,
					%xlabel = \(u\),
					%ylabel = {\(MTE\)},
					]
					%Below the red parabola is defined
					\addplot [
					thick,
					domain=0:7/16, 
					samples=100, 
					color=black,
					]
					{1};
					\addplot [
					thick,
					domain=7/16:1/2, 
					samples=100, 
					color=black,
					]
					{2*x + 2/16};
					\addplot [
					thick,
					domain=1/2:5/8, 
					samples=100, 
					color=black,
					]
					{-2*x + 2 + 2/16};
					\addplot [
					thick,
					domain=5/8:11/16, 
					samples=100, 
					color=black,
					]
					{2*x -6/16};
					\addplot [
					thick,
					domain=11/16:1, 
					samples=100, 
					color=black,
					]
					{1};
					\addlegendentry{$f_2$};
					
					\addplot[dashed] coordinates{(0.5-1/16,0) (0.5-1/16,1)};
					\addplot[dashed] coordinates{(0.5+1/16,0) (0.5+1/16,1)};
					\addplot[] coordinates{(0.5,0.5)} node[] {$2\delta$};
					\addplot[dashed] coordinates{(0.5+3/16,0) (0.5+3/16,1)};
					\addplot[] coordinates{(0.5+2/16,0.5)} node[] {$2\delta$};
					
					\addplot[dashed] coordinates{(0,1+2/16) (1/2 ,1+2/16)};
					\addplot[] coordinates{(0.25,1+1/16)} node[] {$\delta$};

					%after end axis/.code={
					%	\path (axis cs:0,0) node [anchor=north west,yshift=-0.075cm,xshift=-0.075cm] {0};
					%}
					
				\end{axis}

			\end{tikzpicture} 
\end{figure}

    The density $f_1$ of the uniform distribution is continuous and continuously differentiable. It is well-known (and easy to verify) that the uniform distribution is Myerson regular. Therefore, $f_1 \in \mathcal{F}$. The perturbed density $f_2$, by construction, is continuous and almost everywhere continuously differentiable. The upper Dini derivative of $f_2$ belongs to the set $\{-1,0,1\}$. Therefore, to show that $f_2 \in \mathcal{F}$, we only need to verify that it is Myerson regular. We prove this fact in Lemma \ref{lm:f_2_Myerson}. For these two density functions, the difference between their respective evaluation functionals is $|f_1(0.5)-f_2(0.5)| = \delta$. The Hellinger distance can be bounded as follows. Define the function $\Psi(t)=\sqrt{1+t}$. Its second-order derivative
is bounded when $|t|<1/2$; that is, $\sup_{|t|<1/2}|\Psi''(t)|\leq\frac{\sqrt{2}}{2}$. Since $f_{1}(y)=1$, we have 
\begin{align*}
H(f_{1}\|f_{2})^{2}/2 =1-\int_{0}^{1}\Psi\left(\delta\phi\left(\frac{v-1/2}{\delta}\right)\right)dv =\int_{0}^{1}\Psi(0)-\Psi\left(\delta\phi\left(\frac{v-1/2}{\delta}\right)\right)dv.
\end{align*}
By the second-order Taylor expansion, we have 
\begin{align*}
 \Psi(0)-\Psi\left(\delta\phi\left(\frac{v-1/2}{\delta}\right)\right)
\leq -\Psi'(0)\delta\phi\left(\frac{v-1/2}{\delta}\right)+\frac{\sqrt{2}}{4}\delta^{2}\phi^{2}\left(\frac{v-1/2}{\delta}\right).
\end{align*}
By the construction of $\phi$, we have $\int_{0}^{1}\phi\left(\frac{v-1/2}{\delta}\right)dv=0$. By the change of variables $u=(v-1/2)/\delta$, we have 
\begin{align*}
\int_{0}^{1}\phi^{2}\left(\frac{v-1/2}{\delta}\right)dy=\delta\int_{\mathbb{R}}\phi^{2}\left(u\right)du\leq4\delta\int_{-1}^{0}(t+1)^{2}dt=\frac{4}{3}\delta.
\end{align*}
Combining these results together, we obtain a bound on the Hellinger
distance: 
\begin{align*}
H(f_{1}\|f_{2})^{2}\leq\frac{2\sqrt{2}}{3}\delta^{3}.
\end{align*}
Now we set can $\delta = (\frac{3}{8\sqrt{2}n})^{1/3}$, which guarantees that $H(f_{1}\|f_{2})^{2}\leq1/(4n)$. This implies that $\omega\left(1/(2\sqrt{n})\right) \geq (\frac{3}{8\sqrt{2}n})^{1/3}$. Therefore, we obtain the following minimax lower bound:
\begin{align*}
        \inf_{\tilde{f}} \sup_{f \in \mathcal{F}} \mathbb{E}_f |\tilde{f}(v) - f(v)| \geq \frac{1}{8} \omega\left(1/(2\sqrt{n})\right) \geq \frac{1}{8} \left(\frac{3}{8\sqrt{2}} \right)^{1/3} n^{-1/3}.
    \end{align*}
    
\end{proof}

\begin{lemma} \label{lm:f_2_Myerson}
    For $\delta>0$ sufficiently small, the density function $f_2$ defined by (\ref{eq:f_2}) is Myerson regular. 
\end{lemma}

\begin{proof}[Proof of Lemma \ref{lm:f_2_Myerson}]
    The density function $f_2$ can be written as a piecewise function:
    \begin{align*}
    f_2(v)=
    \begin{cases}
    1, & \text{ if }v\in[0,1/2-\delta) \cup [1/2+3\delta,1],\\
    v+\frac{1}{2}+\delta, & \text{ if }v\in[1/2-\delta,1/2),\\
    -v+\frac{3}{2}+\delta, & \text{ if }v\in[1/2,1/2+2\delta),\\
    v+\frac{1}{2}-3\delta, & \text{ if }v\in[1/2+2\delta,1/2+3\delta).
    \end{cases}
    \end{align*}
    We want to show that the function $(1-F_2)^{-1}$ is convex, where $F_2$ is the cumulative distribution function of $f_2$. Because the derivative of $(1-F_2)^{-1}$ is continuous and piecewise differentiable, we just need to show that the second-order derivative of $(1-F_2)^{-1}$ is piecewise nonnegative. This is equivalent to checking that the function $\psi \equiv 2f_2^2 + (1-F_2)f'_2$, which is the numerator of the second-order derivative of $(1-F_2)^{-1}$, is piecewise nonnegative.
    
    We examine one by one the four regions in the piecewise definition of $f_2$. On the first region $[0,1/2-\delta) \cup [1/2+3\delta,1]$, the density $f_2$ is equal to the density of the uniform distribution and therefore is Myerson regular. By elementary calculations, on the remaining three intervals, the functions $F_2$ and $\varphi$ are equal to
    \begin{align*}
        F_2(v) = 
        \begin{cases}
        v^2/2 + (1/2 + \delta)v + (\delta - 1/2)^2/2, & v \in [1/2-\delta,1/2), \\
        -v^2/2 + (3/2 + \delta)v + (\delta^2 - \delta - 1/4)/2, & v \in [1/2,1/2+2\delta), \\
        v^2/2 + (1/2 - 3\delta)v + (-3\delta + 1/2)^2/2, & v \in [1/2+2\delta,1/2+3\delta),
        \end{cases}
    \end{align*}and
    \begin{align*}
        \psi(v) = 
        \begin{cases}
        3(v+1/2 +\delta)^2/2 + 1+\delta, & v \in [1/2-\delta,1/2), \\
        3(-v+3/2 +\delta)^2/2 + \delta(1+\delta), & v \in [1/2,1/2+2\delta), \\
        3(v- 1/2 +3\delta)^2/2 + 1-3\delta, & v \in [1/2+2\delta,1/2+3\delta),
        \end{cases}
    \end{align*}
    respectively. Therefore, the function $\varphi$ is bounded away from zero when $\delta < 1/3$. This proves that $(1-F_2)^{-1}$ is convex. Hence, $f_2$ is Myerson regular.
\end{proof}

\bibliographystyle{chicago}
\bibliography{references.bib}

\end{document}